\documentclass[conference,a4paper,10pt]{IEEEtran}
\usepackage[left=1.5cm,right=1.5cm,top=2.1cm,bottom=4.5cm]{geometry}
\IEEEoverridecommandlockouts

\usepackage[utf8]{inputenc} 
\usepackage[T1]{fontenc}
\usepackage{url}
\usepackage{ifthen}
\usepackage{cite}
\usepackage[cmex10]{amsmath} 
\usepackage{amssymb}
\usepackage{graphicx}
\usepackage{epstopdf}
\usepackage{epsfig}
\usepackage{multirow}
\usepackage{amsthm}
\usepackage{comment}
\usepackage{caption}
\usepackage{enumitem}
\usepackage{amsmath}
\usepackage{color}
\usepackage{bbm}
\usepackage{algorithm}
\usepackage{algpseudocode}
\usepackage{csquotes}
\usepackage{subfiles}
\usepackage{caption, multirow, makecell}
\usepackage{array}
\usepackage{caption}

\newtheorem{thm}{Theorem}

\newtheorem{example}{Example}

\newtheorem{defn}{Definition}

\newtheorem{rem}{Remark}

\def\BibTeX{{\rm B\kern-.05em{\sc i\kern-.025em b}\kern-.08em
    T\kern-.1667em\lower.7ex\hbox{E}\kern-.125emX}}

\begin{document}

\title{Coded Caching with Shared Caches from Generalized Placement Delivery Arrays\\

\thanks{This work was supported partly by the Science and Engineering Research Board (SERB) of Department of Science and Technology (DST), Government of India, through J.C Bose National Fellowship to B. Sundar Rajan.
}
}

\author{\IEEEauthorblockN{Elizabath Peter and B. Sundar Rajan}
	\IEEEauthorblockA{\textit{Department of Electrical Communication Engineering} \\
		\textit{Indian Institute of Science}\\
		Bangalore, India \\
		\{elizabathp,bsrajan\}@iisc.ac.in}
}

\maketitle

\begin{abstract}
We consider the coded caching problem with shared caches where several users share a cache, but each user has access to only a single cache. For this network, the fundamental limits of coded caching are known for centralized and decentralized settings under uncoded placement. In the centralized case, to achieve the gains offered by coded caching, one requires a sub-packetization which increases exponentially with the number of caches. The dedicated cache networks had a similar issue, and placement delivery arrays (PDAs) were introduced as a solution to it. Using the PDA framework, we propose a procedure to obtain new coded caching schemes for shared caches with lower sub-packetization requirements. The advantage of this procedure is that we can  transform all the existing PDA structures into coded caching schemes for shared caches, thus resulting in low sub-packetization schemes.
We also show that the optimal scheme given by Parrinello, Ünsal and Elia (Fundamental Limits of Coded Caching with Multiple Antennas, Shared Caches and Uncoded Prefetching) can be recovered using a Maddah-Ali Niesen PDA.

\end{abstract}


\section{Introduction}
The accelerated growth in the number of users and their rising demand for video-on-demand services have a multiplicative effect on the peak hour data traffic. Caching has been proposed as a promising tool to shift some of the peak hour traffic to off-peak times by utilizing the memories distributed across the network. Coded caching problem consists of two phases: \textit{placement (or prefetching) phase} and \textit{delivery phase} \cite{MaN}. In the placement phase, the network is not congested, and caches are populated with portions of file contents. The placement is independent of the future demands of the users. The delivery phase starts when the users' demands are informed to the server, and the main limitation here is the transmission load (number of bits transmitted) over the shared link that is required to serve the users. 

Maddah-Ali and Niesen introduced the idea of coded caching  in \cite{MaN}, which pointed out the need for the joint design of placement and delivery phases in a caching system. The authors considered a network with a single server having access to a library of $N$ equal-length files connected to $K$ users through a shared error-free link. Each user possesses a memory of size equal to $M$ files. For this network, the scheme in \cite{MaN} (which we refer to as MN scheme henceforth) is shown to be optimal for distinct demands under uncoded prefetching \cite{YMA}. The coded caching problem has then been extended in several directions, such as decentralized caching \cite{MaN2}, online caching \cite{PMN} and many other settings as well. 

In \cite{PUE}, the authors considered a setting where the server is equipped with multiple antennas and is connected to $K$ users assisted by $\Lambda \leq K$ helper caches. Each cache serves an arbitrary number of users, but each user can access only a single cache. For this setting, an optimal centralized coded caching scheme based on uncoded placement was proposed \cite{PUE}. We presently focus only on the single-antenna case and refer to the corresponding scheme as the PUE scheme. For the MN scheme and the PUE scheme, the sub-packetization level which is defined as the number of packets to which a file is split into, grows exponentially with respect to the number of caches. Therefore, it is desirable to have schemes with lower sub-packetization level which will make them suitable for practical realization. In light of the above, for a dedicated cache network, Yan \textit{et al.} \cite{YCT} introduced the idea of Placement Delivery Array (PDA) which could provide schemes with low sub-packetization levels, and characterized the placement and delivery phases in a single array. In this work, we make use of the PDA structures to identify new coded caching schemes for shared caches having lesser sub-packetization constraints than the PUE scheme. 

\subsection{Related Work}
Coded caching with shared caches was first addressed in \cite{MaN2}, which considered a setting where each cache serves an equal number of users and proposed a decentralized caching scheme for the same. In \cite{PUE}, the centralized setting is considered and an optimal coded caching scheme is proposed for distinct demand case. For non-distinct demands, an improved delivery scheme was introduced in  \cite{KDTR}. A caching scheme based on coded placement was proposed in \cite{IZY}, which outperforms the scheme in \cite{PUE} in certain memory regimes. An optimal decentralized coded caching scheme for shared caches was proposed in \cite{DuTh}. An alternate optimal delivery scheme for the decentralized setting was given in \cite{PR}. 

For dedicated cache network, several coded caching schemes with lower sub-packetization have been introduced. The relationship with sub-packetization and transmission load was first discussed by Shanmugam \textit{et al.} in \cite{SJTLD} and proposed a scheme based on grouping method. \cite{YCT} was the first to introduce the concept of PDA to describe caching schemes, and presented two new constructions of PDA that would result in schemes with lesser number of file division than the MN scheme. \cite{YCT} also showed that MN scheme can be represented using a PDA, which in fact belongs to the class of regular PDA. In regular PDA, each integer occurs the same number of times. Various coded caching schemes based on the design of PDA were proposed recently \cite{PDA1} - \cite{PDA6}. Also, there exists several other constructions that focus on reducing sub-packetization, some of which uses resolvable block design and linear block codes \cite{TaR}, the strong edge colouring of subset graphs \cite{YTCC}, projective space \cite{Pkrishnan} etc. Recently, a new class of PDA termed as consecutive and $t$-cyclic PDA were introduced \cite{SaR} to describe the delivery scheme for a special class of multi-access coded caching problem.
\\
\\
\noindent \textit{Notations}: For any integer $n$, $[n]$ denotes the set $\{1,2,...,n\}$. For any set $\mathcal{S}$, $|\mathcal{S}|$ denotes the cardinality of $\mathcal{S}$. Binomial coefficients are denoted by $\binom{n}{k}$, where $\binom{n}{k} \triangleq \frac{n!}{k!(n-k)!}$ and $\binom{n}{k}$ is zero for $n < k$.
Bold uppercase and lowercase letters are used to denote matrices and column vectors, respectively. The contents stored in each helper cache, $\lambda$ is denoted by $\mathcal{Z}_{\lambda}$. An identity matrix of size $n$ is denoted as $\mathbf{I}_{n}$. The set of positive integers is denoted by $\mathbb{Z}^{+}$. 

\subsection{Our Contributions}
Our contributions are summarized below:
\begin{itemize}
	\item A procedure is proposed to obtain new coded caching schemes for shared cache from PDAs. The advantage of this is that we can transform all the existing PDA structures into coded caching schemes for shared caches, thus resulting in low sub-packetizaton schemes.
	
	\item We show that the optimal scheme in \cite{PUE} (PUE scheme) can be recovered using a MN PDA.
\end{itemize}

\section{System Model and Background}
In this section, we first discuss the shared caching problem, and then briefly review the definition of PDA and see how it represents a coded caching scheme.

\begin{figure}[t!]
	\begin{center}
		\captionsetup{justification=centering}
		\includegraphics[width=\columnwidth]{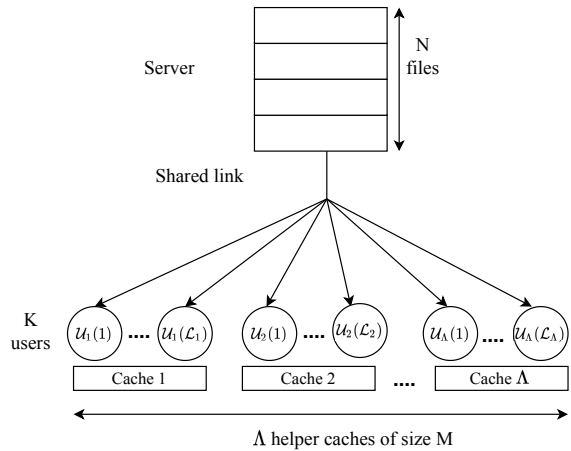}
		\caption{Problem setting considered in this work.}
		\label{fig:setting}
	\end{center}
\end{figure}

\subsection{Setting}
We consider a shared cache network as illustrated in Fig~\ref{fig:setting}. There is a server with $N$ equal-length files $\mathcal{W}=\{W^1, W^2,\ldots, W^N\}$, connected to $K$ users through a shared error-free link. There are $\Lambda\leq K$ helper caches, each of size equal to $M$ files, where $M \in [0,N]$. Each user gets access to one of the helper caches and there is no limit on the number of users served by each cache. We call this problem as a $(\Lambda, K, M, N )$ shared caching problem.
 
The $(\Lambda, K,M,N)$ shared caching problem operates in three phases:\\
\textit{a) Placement phase}: In this phase, the server fills the caches with parts of the file contents, satisfying the memory constraint. This is carried out without knowing the future demands of the users and their association to the caches.\\
\textit{b) User-to-cache association phase:} The placement phase is followed by an additional phase in which each user gets connected to one of the helper nodes from which it can download the contents at zero cost. The set of users assigned to cache $\lambda \in [\Lambda]$ is denoted as $\mathcal{U}_{\lambda}$, and all these disjoint sets together form a partition of the set of $K$ users. The overall association of the users to the caches is represented as:
\begin{equation*}
\mathcal{U} = \{\mathcal{U}_1, \mathcal{U}_2,\ldots, \mathcal{U}_{\Lambda}\}.
\end{equation*}
This assignment of users to caches is independent of the cached contents and the subsequent demands. For any user-to-cache association $\mathcal{U}$, the association profile $\mathcal{L}$ describes the number of users served by each cache. Therefore,
\begin{equation*}
\mathcal{L} = \{\mathcal{L}_1, \ldots, \mathcal{L}_{\Lambda}\}
\end{equation*}
where $\mathcal{L}_i = |\mathcal{U}_i|$. Without loss of generality, assume that $\mathcal{L}_i \geq \mathcal{L}_j$ $\forall i\leq j$ and $\mathcal{U}_i$ to be an ordered set.
Several user-to-cache associations result in the same $\mathcal{L}$. Therefore, each $\mathcal{L}$ represents a class of $\mathcal{U}$. \\
\textit{c) Delivery phase:} Users inform their demands to the server. Let the index of the file demanded by the user $k \in [K]$ be denoted as $d_k$. On receiving the demand vector $\mathbf{d}=(d_1,d_2,\ldots,d_K)$, the server broadcasts a message over the shared link to the users which would enable each of them to decode their requested files. The aim of the server is to design the placement and delivery scheme so as to minimize the number of transmissions required to serve the users. 

Let $R(\mathcal{L},\mathbf{d})$ denote the transmitted load (normalized by the file size) over the shared link to satisfy the demand $d$ when the association profile is $\mathcal{L}$. The worst-case load corresponds to $\underset{\mathbf{d}\in\{1,\ldots,N\}^{K}}{\textrm{max}}R(\mathcal{L},\mathbf{d})$. The optimal worst-case load for a given $\mathcal{U}$ with association profile $\mathcal{L}$ is denoted by $R^{*}(\mathcal{L})$, which is the minimum achievable worst-case load over all possible caching and delivery schemes.

\subsection{Placement Delivery Array (PDA)}
In this part, we consider the dedicated cache network which consists of a server having access to $N$ equally-sized files, connected to $K$ users. Each user is endowed with a cache of size equal to $M$ files.
\begin{defn}
	(\cite{YCT}) For positive integers $K, F, Z$ and $S$, an $F \times K$ array $\mathbf{P}=(p_{j,k})$, $j \in [0,F)$ and $k \in [0,K)$, composed of a specific symbol $\star$ and $S$ non-negative integers $0,1,\ldots, S-1$, is called a $(K,F,Z,S)$ placement delivery array (PDA) if it satisfies the following three conditions: \\
	\textit{C1}. The symbol $\star$ appears $Z$ times in each column.\\
	\textit{C2}. Each integer occurs at least once in the array.\\
	\textit{C3}. For any two distinct entries $p_{j_1,k_1}$ and $p_{j_2,k_2}$,\\ $p_{j_1,k_1}=p_{j_2,k_2}=s$ is an integer only if
	\begin{enumerate}[label=(\alph*)]
		
		\item $j_1 \neq j_2$, $k_1 \neq k_2$, i.e., they lie in distinct rows and distinct columns, and
		\item $p_{j_1,k_2}=p_{j_2,k_1}=\star$, i.e., the corresponding $2\times2$ sub-array formed by rows $j_1, j_2$ and columns $k_1,k_2$ must be of the following form:\\
		\begin{center}
			$\begin{pmatrix}
			s & \star\\
			\star & s
			\end{pmatrix}$
			\hspace{0.3cm}or\hspace{0.3cm}
			$\begin{pmatrix}
			\star & s \\
			s & \star
			\end{pmatrix}$ 
			
		\end{center}
		
	\end{enumerate}
\end{defn}

\begin{thm}
	(\cite{YCT}) For a given $(K, F, Z, S)$ PDA $\mathbf{P}=(p_{j,k})_{F \times K}$, a $(K,M,N)$ coded caching scheme can be obtained with sub-packetization $F$ and $\frac{M}{N}=\frac{Z}{F}$ using Algorithm 1. For any demand $\mathbf{d}$, the demands of all the users are met with a transmission load of $R=\frac{S}{F}$.
\end{thm}

\begin{algorithm}
	\renewcommand{\thealgorithm}{1}
	\caption{Coded caching scheme based on PDA \cite{YCT}}
	\begin{algorithmic}[1]
		\Procedure{Placement}{$\mathbf{P},\mathcal{W}$}       
		\State Split each file $W^n$ in $\mathcal{W}$ into $F$ packets: $W^n =\{W^n_j: j \in [0,F)\}$
		\For{\texttt{$k \in 0,1,\ldots,K-1$}}
		\State  $\mathcal{Z}_k$ $\leftarrow$ $\{W^n_{j}: p_{j,k}=\star, \forall n \in [N]\}$
		\EndFor
		\EndProcedure
		
		\Procedure{Delivery}{$\mathbf{P},\mathcal{W},\mathbf{d}$} 
		\For{\texttt{$s = 0,1,\ldots, S-1$}}
		\State Server sends $\underset{p_{j,k}=s, j\in [0,F),k\in[0,K)}{\bigoplus}W^{d_k}_j$
		\EndFor    
		\EndProcedure
	\end{algorithmic}
\end{algorithm}

In a $(K,F,Z,S)$ PDA $\mathbf{P}$, the rows represent packets and the columns represent users. For any $k \in [0,K)$ if $p_{j,k}=\star$, then it implies that the user $k$ has access to the $j^{th}$ packet of all the files. If $p_{j,k}=s$ is an integer, it means that user $k$ does not have access to the $j^{th}$ packet of any of the files. Condition $C1$ guarantees that all users have access to some $Z$ packets of all the files. According to the delivery procedure in Algorithm $1$, the server sends a linear combination of the requested packets indicated by the integer $s$ in the PDA. Therefore, condition $C2$ implies that the number of messages transmitted by the server is exactly $S$ and the transmitted load is $\frac{S}{F}$. Condition $C3$ insures the decodability condition.

The following theorem shows that MN scheme can be obtained from a PDA.
\begin{thm}
	(\cite{YCT}) For a $(K,M,N)$ caching system with $\frac{M}{N} \in \{0, \frac{1}{K}, \frac{2}{K}, \ldots, 1 \}$ and let $t=\frac{KM}{N}$, there exists a $t+1$-$(K,F,Z,S)$ PDA with $F=\binom{K}{t}$, $Z=\binom{K-1}{t-1}$ and $S=\binom{K}{t+1}$.
\end{thm}

MN scheme can be represented using a regular PDA. In MN scheme, choose $t \in [0,K]$ such that $F = \binom{K}{t}$. Each row of the PDA is indexed by sets $\mathcal{T} \subseteq [K]$ where $|\mathcal{T}|=t$. Each user stores subfiles $W^{n}_{\mathcal{T}}, \forall n \in [N]$ if $k \in \mathcal{T}$, thus $Z=\binom{K-1}{t-1}$. In the delivery phase of the scheme, there are $\binom{K}{t+1}$ transmissions. Therefore, $\binom{K}{t+1}$ distinct integers are required to represent each of those transmissions resulting in $S = \binom{K}{t+1}$. To fill the PDA with these integers, a bijection $f$  is defined from the $t+1$ sized subsets of $\{1,2,\ldots,K\}$ to $[0,S)$ such that 
\begin{equation}
p_{\mathcal{T},k} = 
\begin{cases}
f(\mathcal{T}\cup \{k\}), & \textrm{if \hspace{0.05cm}} k \notin \mathcal{T}.\\
\star, & \textrm{elsewhere}.
\end{cases}
\label{pda}
\end{equation}

From the above expression, it is evident that each integer appears exactly $t+1$ times, hence the PDA obtained belongs to a $t+1$- regular PDA. The PDA which represents the MN scheme is referred as MN PDA henceforth.

\section{A New Class of Coded Caching schemes for Shared Caches using PDAs }

In this section, we first formally define a modified placement delivery structure called Generalized PDA, which can completely characterize all the three phases in a shared caching problem in a single array. We then describe the procedure to obtain new coded caching schemes for shared caches using the PDA framework which led to the design of Generalized PDAs.


\subsection{Generalized PDA}
\begin{defn}
	For positive integers $K,F,Z,S$ and $I$, an $F \times K$ array $G=(g_{j,k})$, $j \in [0,F)$, $k \in [0, K)$ composed of $\star $ and integer entries from $S^{(I)}$ where  
	$S^{(I)}=\{s^{(i)} : s \in [0,S), i\in [I]\}$ is called a $(K,F,Z,S^{(I)})$ generalized PDA if it satisfies the following conditions:
	\begin{enumerate}[label= C\arabic*.]
		\item The symbol $\star$ appears $Z$ times in each column.
		\item Each integer and each superscript occurs at least once in the array.
		\item For any two distinct entries $g_{j_1,k_1}$ and $g_{j_2,k_2}$, $g_{j_1,k_1}=g_{j_2,k_2}=s^{(i)}$ is a numerical entry only if
		\begin{enumerate}[label=\alph*)]
			
			\item $j_1 \neq j_2$, $k_1 \neq k_2$, i.e., they lie in distinct rows and distinct columns, and
			\item $g_{j_1,k_2}=g_{j_2,k_1}=\star$, i.e., the corresponding $2\times2$ sub-array formed by rows $j_1, j_2$ and columns $k_1,k_2$ must be of the following form:\\
			\begin{center}
				$\begin{pmatrix}
				s^{(i)} & \star\\
				\star & s^{(i)}
				\end{pmatrix}$
				\hspace{0.3cm}or\hspace{0.3cm}
				$\begin{pmatrix}
				\star & s^{(i)} \\
				s^{(i)} & \star
				\end{pmatrix}$ 
				
			\end{center}
		\end{enumerate}
		\item For any four distinct entries $g_{j_1,k_1}$, $g_{j_1,k_2}$, $g_{j_2,k_1}$ and $g_{j_2,k_2}$, if $g_{j_1,k_1} =s^{(i_1)}$, $g_{j_2,k_1}=\star$ and  $g_{j_1,k_2}=s^{(i_2)}$, then $g_{j_2,k_2}=\star$.

	\end{enumerate}
	
\end{defn}

\subsection{Description of the scheme} 

For a given $(\Lambda,K,M,N)$ shared caching problem, we start with a $(\Lambda, F, Z, S)$ PDA $\mathbf{P} = (p_{j,\lambda})$ where $j \in [0,F)$ and $\lambda \in [0,\Lambda)$ with $\frac{M}{N}=\frac{Z}{F}$. The parameter $F$ determines the sub-packetization level in the resulting scheme. Each column in $\mathbf{P}$ represents a cache and the caches' are numbered from $[0,\Lambda)$.

Algorithm $2$ explains how to obtain a coded caching scheme for shared cache using the PDA $\mathbf{P}$. Some remarks on Algorithm $2$ are summarized here: The `$\star$'s in $\mathbf{P}$ indicate the contents cached in each cache. To all the integer entries in $\mathbf{P}$, associate `$1$' as a superscript (line $8$). The user-to-cache association $\mathcal{U}$ determines the side-information possessed by each user. Therefore, once $\mathcal{U}$ and $\mathcal{L}$ are known, we construct an array $\mathbf{G}_{F\times K}$ from $\mathbf{P}$ as follows: each column  $\mathbf{p}_{_\lambda}$, $\lambda \in [0,\Lambda)$  is repeated $\mathcal{L}_{\lambda}$ times. In each replication of $\mathbf{p}_{_\lambda}$, the superscript associated with the integer entries in $\mathbf{p}_{_\lambda}$ gets incremented by one. The array $\mathbf{G}$ obtained from the user-to-cache association procedure is a generalized PDA, which is defined in Definition $2$.

\begin{algorithm}
	\renewcommand{\thealgorithm}{2}
	\caption{Coded caching scheme for a $(\Lambda,K,M,N)$ shared cache network using PDA $\mathbf{P}_{F \times \Lambda}$ with $\frac{M}{N}=\frac{Z}{F}$. }
	\begin{algorithmic}[1]
		
		\Procedure{Placement}{$\mathbf{P},\mathcal{W}$}       
		\State Split each file $W^n$ in $\mathcal{W}$ into $F$ packets: $W^n =\{W^n_j: j \in [0,F)\}$
		\State $u \gets 1$
		\For{\texttt{$\lambda \in [0,\Lambda)$}}
		\For{\texttt{$j \in [0,F)$}}
		\State  $\mathcal{Z}_{\lambda} \gets $ $\{W^n_{j}: p_{j,\lambda}=\star, \forall n \in [N]\}$
		
		\If {$p_{j,\lambda} \neq \star$ } 
		\State Associate the superscript $u$ to $p_{j,\lambda}$, i.e, $p_{j,\lambda} = p_{j,\lambda}^{(u)}$
		\EndIf
		\EndFor
		\EndFor
		\EndProcedure
		
		\Procedure{User-to-cache association}{$\mathbf{P},\mathcal{U}$}
		\State Obtain $\mathcal{L} \gets (\mathcal{L}_0, \mathcal{L}_1,\ldots, \mathcal{L}_{\Lambda-1})$, $\mathcal{L}_i \geq \mathcal{L}_j$ $\forall i \leq j$ 
		\State Construct $\mathbf{G} = (g_{j,k}),$ $j \in [0,F)$, $k \in [0,K)$
		\State $k \gets 0$ 
		\For {$\lambda \in [0,\Lambda)$}
		\If {$\mathcal{L}_{\lambda}>0$}
		\For {$i \in [0,\mathcal{L}_{\lambda})$}
		\State $\mathbf{g}_k = \mathbf{p}_{\lambda}^{(u+i)}$ where $\star^{(u+i)}=\star$
		\State $k \gets k+1$		   		
		\EndFor
		\EndIf
		\EndFor
		\EndProcedure

		\Procedure{Delivery}{$\mathbf{G},\mathcal{W},\mathbf{d}$} 
		\For{\texttt{$s \in [0,S)$}}
		\For{$i \in \{1,2,\ldots,\mathcal{L}_0\}$}
		\If {$s^{(i)}$ exists}
		\State Server sends $\underset{\substack{g_{j,k}=s^{(i)}\\ j \in [0,F),k\in[0,K)}}{\bigoplus}W^{d_k}_j$
		\EndIf
		\EndFor
		\EndFor    
		\EndProcedure
	\end{algorithmic}
\end{algorithm}

A $(K, F, Z, S^{(I)})$ generalized PDA $\mathbf{G}_{F\times K}$ represents an $F$-division caching scheme for a $(\Lambda,K,M,N)$ shared cache system with $\frac{M}{N}=\frac{Z}{F}$. 

In $\mathbf{G}$, the columns correspond to the users and the rows correspond to the subfiles. If an entry $g_{j,k}$ is $\star$, it means that the $k^{th}$ user has access to  $j^{th}$ subfile of all the $N$ files. Condition $C1$ ensures that all users have access to same number of sub-files. Thus, $\mathbf{G}$ can represent the side-information possessed by each user.

To obtain the user-to-cache association, $\mathcal{U}$ from $\mathbf{G}$, a grouping process need to be followed. Initially choose an integer $s \in [0,S)$, then for each row find all the columns in which $s$ appears irrespective of its superscript. The so obtained set of columns from each row correspond to a set $\mathcal{U}_{\lambda}$ in $\mathcal{U}$. Condition $C4$ guarantees that these set of users possess the same side-information set. To identify the remaining $\mathcal{U}_{\lambda}$, repeat the above procedure for other integers. 

Next, we describe how $G$ represents the delivery phase. Let $d=(d_0,d_1,\ldots,d_{K-1})$ be the demand vector.
In $\mathbf{G}$, assume that there are $m$ integer entries such that $g_{j_1,k_1}=\ldots=g_{j_m,k_m}=s^{(i)}$ where $s \in [0,S)$ and $i \in [\mathcal{L}_0]$. Then by condition $C3$ in Definition $2$, the sub-array formed by rows $j_1, \ldots, j_m$ and columns $k_1,\ldots,k_m$ is equivalent to a scaled identity matrix $\mathbf{I}_{m}$ up to row or column permutations as follows. 

\begin{table*}[ht]
	\centering
	\caption{Summary of some known PDA constructions}
	\begin{tabular}{| c | c | c | c | c| }
		\hline
		\rule{0pt}{3.5ex}
		\makecell{Schemes and parameters} & \makecell{Number of users\\ $K$} & \makecell{Caching ratio \\ $\frac{M}{N}$} & \makecell{Sub-packetization \\ $F$} & \makecell{Number of integers\\ $S$}  \\ [2pt]
		\hline
		\rule{0pt}{3.5ex}
		\makecell{MN PDA \cite{YCT} \\ For $t$, $K \in \mathbb{Z}^{+}$ with $t \in [0,K)$} & $K$ & $\frac{t}{K}$ & $\binom{K}{t}$ & $\binom{K}{t+1}$ 
		
		\\
		\hline
		\rule{0pt}{2.6ex}
		\multirow{2}{*}{\makecell{Scheme in \cite{YCT} \\ For $m$ and $q$ $\in \mathbb{Z}^{+}$, $q \geq 2$}} & \multirow{2}{*}{$q(m+1)$} & $\frac{1}{q}$ & $q^m$ & $q^{m}(q-1)$\\
		
		\cline{3-5}
		\rule{0pt}{2.6ex}
		& & $\frac{q-1}{q}$ & $(q-1)q^m$ & $q^m$\\
		\hline
		
		\rule{0pt}{5ex}
		\makecell{Scheme in \cite{PDA2}\\ For any $q$, $z$ and $m$ $\in \mathbb{Z}^{+}$\\ with $q \geq 2$ and $z < q$} & $q(m+1)$ & $\frac{z}{q}$ & $\left \lfloor{\frac{q-1}{q-z}} \right \rfloor q^m$ & $(q-z)q^m$\\
		\hline
	\end{tabular}
	
	\label{tab:pda}
\end{table*}

\begin{center}
	
	$  
	\begin{pmatrix}
	s^{(i)} & \star & \ldots  &  \star \\
	\star &  s^{(i)} & \ldots & \star   \\
	\vdots & \vdots & \ddots & \vdots \\
	\star & \star & \ldots & s^{(i)}
	\end{pmatrix}_{m \times m}$
\end{center}

From the above structure, it is evident that the users corresponding to the columns $k_1, \ldots, k_m$ form a clique. Hence, the server sends a message of the form
\begin{equation}
\underset{1 \leq u \leq m}{\bigoplus} W^{d_{k_u}}_{j_u}.
\label{gpda}
\end{equation}

The delivery scheme in line $30$ of Algorithm $2$ follows from \eqref{gpda}. 
Based on the above observations, we arrive at our main result which is stated in the following theorem.

\subsection{Main Result}
\begin{thm}
	\label{new}
	For a given $(\Lambda, K, M, N)$ shared caching problem with an association profile $\mathcal{L}=(\mathcal{L}_0,\ldots,\mathcal{L}_{\Lambda-1})$, an $F$-division caching scheme can be derived from a $(\Lambda, F,Z,S)$ PDA $\mathbf{P}=(p_{j,\lambda})_{F \times \Lambda}$ where $\frac{M}{N}=\frac{Z}{F}$. The delivery load required in the worst-case is 
	\begin{equation}
	R(\mathcal{L}) = \frac{\displaystyle\sum_{s=0}^{S-1} max\{i : s^{(i)} \textrm{ appears in } \mathbf{G}, i \in [\mathcal{L}_0]\}}{F} \textrm {\hspace{0.2cm}.}
	\label{thm4}
	\end{equation}
	where	$\mathbf{G}=(g_{j,k})_{F\times K}$ is the generalized PDA obtained from Algorithm $2$ representing the obtained $F$-division  caching scheme.
\end{thm}

\begin{proof}
	The proof of this theorem directly follows from Algorithm $2$. In a $(\Lambda,F,Z,S)$ PDA $\mathbf{P}$, rows represent the subfiles and columns correspond to the caches. The sub-files stored in each cache is defined by the symbol $\star$ present in the column corresponding to it. Condition $C1$ in the definition of PDA guarantees that each cache stores $NZ$ subfiles, where each subfile is of size $\frac{1}{F}$. Thus, the size of each cache is $NZ/F$, which is equal to $M$. When the user-to-cache assignment $\mathcal{U}$ is known, a generalized PDA $\mathbf{G}$ of size $F \times K$ is constructed as described in lines $16$-$24$ of Algorithm 2. Once the server receives the demand vector $\mathbf{d}$, it transmits a linear combination of those sub-files (of the requested files) marked by the entry $s^{(i)}$, where $s \in [0,S)$ and $i \in [\mathcal{L}_0]$. This process happens for all distinct $s^{(i)}$. The range of the superscript, $i$ is different for each $s$ as it depends on the profile $\mathcal{L}$. Thus, we obtain a transmission rate equal to \eqref{thm4}. The condition $C3$ in Definition $2$ ensures the decodability at each user. 
\end{proof}
\vspace{-0.1cm}
Now, we present an example to describe Theorem 3.\vspace{-0.2cm}
\begin{example} (4,8,4,8) shared caching problem with $\mathcal{L}=(3,2,2,1)$.
\end{example}

Consider a case with $K=8$ users sharing $\Lambda=4$ helper caches, each of size $M=4$ units of file, storing contents from a library of $N=8$ equally-sized files $\mathcal{W}=\{W^1, W^2,\ldots, W^8\}$. 

For the above scenario, the PUE scheme requires a sub-packetization of $\binom{\Lambda}{{\Lambda M}/{N}}$ which is equal to $6$ subfiles. Whereas, using a $(4,2,1,2)$ PDA we could find a coded caching scheme for the same setup, which would only need a sub-packetization of $2$. 

\begin{center}
	$ \mathbf{P} =
	\begin{pmatrix}
	\star & 1 & \star  & 0 \\
	0 & \star  & 1 & \star \\
	
\end{pmatrix}_{2 \times 4}$

\end{center}

\noindent Let $\mathbf{P}$ be a $(4,2,1,2)$ PDA satisfying the constraint $\frac{M}{N}=\frac{Z}{F}$. 

\textit{Placement}: As there are $2$ rows in $\mathbf{P}$, each file gets split into $2$ subfiles, i.e, $W^n =\{ W^n_0, W^n_1\}$ $\forall n \in [8].$ Since there is only a single $\star$ in every column, each cache stores one subfile of all the files marked by $\star$ in its corresponding column. We modify $\mathbf{P}$ by attaching the superscript to integer entries.

$$ \mathbf{P} =
\begin{pmatrix}
\star & 1^{(1)} & \star  & 0^{(1)} \\
0^{(1)} & \star  & 1^{(1)} & \star \\

\end{pmatrix}_{2 \times 4}$$

\textit{User-to-cache assignment}: Let $\mathcal{U}=\{\{1,2,3\}, \{4,5\}, \{6,7\},\{8\}\}$ with a profile $\mathcal{L}=(3,2,2,1)$. Then, we construct the generalized PDA $\mathbf{G}$ as 
$$  \mathbf{G} =
\begin{pmatrix}
\star & \star & \star & {1^{(1)}} & {1^{(2)}} & \star & \star & {0^{(1)}} \\
{0^{(1)}} & {0^{(2)}} & {0^{(3)}} & \star  & \star & {1^{(1)}} &{ 1^{(2)}} & \star \\

\end{pmatrix}_{2 \times 8}$$

\textit{Delivery}: Let $d=(1,2,3,4,5,6,7,8)$. There is a transmission corresponding to every integer with a distinct superscript in $\mathbf{G}$. The server takes an XOR of those requested subfiles marked by a particular numerical entry in $\mathbf{G}$.
Thus, the messages transmitted are
\begin{equation*}
\begin{aligned}
X_{0^{(1)}} & = W^1_{1} \oplus W^8_{0} & & &  X_{1^{(1)}} & = W^4_{0} \oplus W^6_{1} \\
X_{0^{(2)}} & = W^2_{1} & & & X_{1^{(2)}} & = W^5_{0} \oplus W^7_{1} \\
X_{0^{(3)}} & = W^3_{1}  \textrm{ \hspace{0.2cm}}. 
\end{aligned}
\end{equation*}

The load required by our scheme is ${5}/{2}$. Whereas, the load achieved by the PUE scheme is ${11}/{6}$ .

Thus, by utilizing the extensive results available on PDA constructions, we could get more practically realizable coded caching schemes for shared caches without paying much in delivery load. In Table I, we list some of the known PDAs which are useful to our discussion. The following examples make use of the PDAs given in Table I to illustrate the sub-packetization reduction achieved by our procedure. 

\begin{example}
$(9,45,15,45)$ shared caching problem with $\mathcal{L}=(8,7,6,6,5,4,4,3,2)$.
\end{example}

There is a server with $N=45$ files connected to $K=45$ users and to $\Lambda = 9$ caches, each with a normalized size $\frac{M}{N}=\frac{1}{3}$. Then for an association profile $\mathcal{L}=(8,7,6,6,5,4,4,3,2)$, sub-packetization and the worst-case delivery load achieved by our PDA-based construction and the PUE scheme are shown in Table II.

\begin{table}[ht]
\centering
\caption{Comparison (Using Construction A in \cite{YCT})}
\begin{tabular}{ c c c }
\hline
\\
Schemes & \makecell{Sub-packetization\\ $F$} & \makecell{Delivery load\\ $R(\mathcal{L})$} \\
\hline
\\

PUE scheme \cite{PUE}	& $84$  & $897/84 \approx 10.68$ \\
\\
Our scheme & 9 & 14 \\
\hline
\end{tabular}
\label{tab:multicol}
\end{table}

\begin{example}
$(9,45,30,45)$ shared caching problem with $\mathcal{L}=(8,7,6,6,5,4,4,3,2)$.
\end{example}
We consider the same setting as in Example $2$ with only change in the normalized cache size. That is, $N=45$, $K=45$, $\Lambda=9$ with $\frac{M}{N}=\frac{2}{3}$. Constructions given in either \cite{YCT}  or \cite{PDA2} can be used to get a PDA satisfying the above parameters. In this example, we employ the PDA from \cite{PDA2} and the performance analyses against the PUE scheme is given below.

\begin{table}[ht]
\centering
\caption{Comparison (Using the PDA in \cite{PDA2})}
\begin{tabular}{ c c c }
\hline
\\
Schemes & \makecell{Sub-packetization\\ $F$} & \makecell{Delivery load\\ $R(\mathcal{L})$} \\
\hline
\\

PUE scheme \cite{PUE}	& $84$  & $279/84 \approx 3.32$ \\
\\
Our scheme  & 18 & 3.5 \\
\hline
\end{tabular}
\label{tab:multicol}
\end{table}

\subsection{Generalized PDA representation for the PUE scheme}
In this subsection, we show that the PUE scheme can also be derived from a PDA.  

Consider a $(\Lambda,K,M,N)$ shared caching problem with an association profile $\mathcal{L}=(\mathcal{L}_0,\ldots,\mathcal{L}_{\Lambda-1})$. Let $t_s \triangleq\frac{\Lambda M}{N}$. To obtain the PUE scheme, we need to start with an MN PDA corresponding to $\Lambda$ users, satisfying the condition $\frac{M}{N}=\frac{Z}{F}$. Therefore, consider a $(\Lambda,F,Z,S)$ MN PDA $\mathbf{P}$ where $F = \binom{\Lambda}{t_s}$, $Z=\binom{\Lambda-1}{t_s-1}$ and $S=\binom{\Lambda}{t_s+1}$.
As mentioned in Section II.B, the rows of $\mathbf{P}$ are indexed by sets $\mathcal{T}\subseteq [0,\Lambda)$, where $|\mathcal{T}|=t_s$. For each column $\mathbf{p}_{_\lambda}$, the symbol $\star$ is present in those rows in which $\lambda \in \mathcal{T}$. The integer entries in $\mathbf{P}$ are obtained by defining a bijective function as described in \eqref{pda}. To all the integer entries, associate the superscript `$1$' and then follow the procedure in lines $13$ - $25$ of Algorithm 2. The so constructed $(K,F,Z,S^{{(\mathcal{L}_0})})$ generalized PDA $\mathbf{G}$ can completely characterize the shared caching problem. We further illustrate this using the Example $1$ discussed in Section III.B.

\textit{Example:}$(4,8,4,8)$ shared caching problem with $\mathcal{L}=(3,2,2,1)$.

$t_s = \Lambda M/N =2$ results in $\binom{\Lambda}{t_s}=6$, $\binom{\Lambda-1}{t_s-1}=3$ and $\binom{\Lambda}{t_s+1}=4$. Therefore, start with a $(4,6,3,4)$ PDA $\mathbf{P}$. In this example, the caches are numbered as $\{1,2,3,4\}$.

\[ \mathbf{P} = 
\begin{pmatrix}
\star & \star & {0}  & {1} \\
\star & {0} & \star &  {2}\\
\star & {1} & {2} & \star \\
{0} & \star & \star & {3}\\
{1} & \star & {3} & \star \\ 
{2}  & {3} & \star & \star
\end{pmatrix}_{6 \times 4}
\Longrightarrow
\begin{pmatrix}
\star & \star & {0}^{(1)}  & {1}^{(1)} \\
\star & {0}^{(1)} & \star &  {2}^{(1)}\\
\star & {1}^{(1)} & {2}^{(1)} & \star \\
{0}^{(1)} & \star & \star & {3}^{(1)}\\
{1}^{(1)} & \star & {3}^{(1)} & \star \\ 
{2}^{(1)}  & {3}^{(1)} & \star & \star
\end{pmatrix}
\]

The rows of $\mathbf{P}$ are represented by 2-sized subsets of $\{1,2,3,4\}$ arranged in lexicographic order. Hence, $W^n = \{W^n_{\{1,2\}}, W^n_{\{1,3\}}, W^n_{\{1,4\}}, W^n_{\{2,3\}}, W^n_{\{2,4\}}, W^n_{\{3,4\}}\}$. The placement given by this PDA is exactly same as that of the placement in \cite{PUE}.

For $\mathcal{U}=\{\{1,2,3\},\{4,5\},\{6,7\},\{8\}\}$, we obtain a $(8,6,3,4^{(3)})$ generalized PDA $\mathbf{G}$ .
$$ \small{  \mathbf{G} =
\begin{pmatrix}
\star & \star & \star & \star & \star & {0^{(1)}} &  {0^{(2)}}& {1^{(1)}} \\
\star & \star & \star & {0^{(1)}} &  {0^{(2)}} & \star & \star & {2^{(1)}} \\
\star & \star & \star & {1^{(1)}} & {1^{(2)}} & {2^{(1)}} & {2^{(2)}} & \star \\
{0^{(1)}} & {0^{(2)}} & {0^{(3)}}& \star & \star & \star & \star & {3^{(1)}}\\
{1^{(1)}}& {1^{(2)}} & {1^{(3)}} & \star & \star & {3^{(1)}} & {3^{(2)}} & \star \\
{2^{(1)}}& {2^{(2)}} & {2^{(3)}} & {3^{(1)}} & {3^{(2)}} & \star & \star & \star  
\end{pmatrix}_{6 \times 8}}$$

\noindent Let $\mathbf{d}=(1,2,3,4,5,6,7,8)$, the delivery is as follows
\begin{equation*}
\begin{aligned}
X_{0^{(1)}} & = W^1_{\{2,3\}} \oplus W^4_{\{1,3\}} \oplus W^6_{\{1,2\}} \\
X_{1^{(1)}} & = W^1_{\{2,4\}} \oplus W^4_{\{1,4\}} \oplus W^8_{\{1,2\}} \\
X_{2^{(1)}} & = W^1_{\{3,4\}} \oplus W^6_{\{1,4\}} \oplus W^8_{\{1,3\}} \\
X_{3^{(1)}} & = W^4_{\{3,4\}} \oplus W^6_{\{2,4\}} \oplus W^8_{\{2,3\}} \\
X_{0^{(2)}} & = W^2_{\{2,3\}} \oplus W^5_{\{1,3\}} \oplus W^7_{\{1,2\}} \\
X_{1^{(2)}} & = W^2_{\{2,4\}} \oplus W^5_{\{1,4\}} \\
X_{2^{(2)}} & = W^2_{\{3,4\}} \oplus W^7_{\{1,4\}} \\
X_{3^{(2)}} & = W^5_{\{3,4\}} \oplus W^7_{\{2,4\}} \\
X_{0^{(3)}} & = W^3_{\{2,3\}}  \textrm{ \hspace{0.01cm},\hspace{0.2cm}}
X_{1^{(3)}}  = W^3_{\{2,4\}} \textrm{ \hspace{0.2cm}and \hspace{0.2cm}} 
X_{2^{(3)}}  = W^3_{\{3,4\}} \textrm{ \hspace{0.01cm}}. 
\end{aligned}
\end{equation*}

The signals transmitted and the load required match exactly with the PUE scheme. Thus, we could represent the PUE scheme using generalized PDA. 


\begin{rem}
Note that the coding gain, which is defined as the number of users benefiting from a single multicast message, is not the same for all transmissions. This in fact arises due to the asymmetry in the user-to-cache assignment.
\end{rem}

\begin{rem}
If the assignment of users to caches is uniform, all the integer entries will be occurring equal number of times in generalized PDA $\mathbf{G}$. Then, the delivery load expression \eqref{thm4} reduces to
\begin{equation}
R(\mathcal{L}) = \frac{K S}{\Lambda F} \textrm {\hspace{0.2cm}.}
\end{equation} 
In this case, every transmission provides the same coding gain.
\end{rem}

\begin{rem}
Throughout this work, we focused on uncoded placement. Hence, we considered only the PUE scheme for comparison. The other schemes mentioned in \cite{IZY}, \cite{XGW} use coded placement to achieve a better performance than the PUE scheme. The schemes given in \cite{IZY} and \cite{PaE}  want the association profile, $\mathcal{L}$ to be known during the placement phase itself. Whereas, our construction and PUE scheme follow an association profile oblivious placement.

\end{rem}

\section{Conclusion}

In this work, we introduced the generalized PDA structure to completely describe the coded caching schemes for shared caches with a single array. We proposed a procedure to derive new coded caching schemes for shared caches using PDAs. This enabled us to use the available PDA constructions to get schemes for shared caches with low sub-packetization levels.

\end{document}